\documentclass[11pt, reqno]{amsart}
\usepackage[a4paper]{geometry}
\usepackage[T1]{fontenc}
\usepackage{inputenx}
\usepackage{setspace}
\usepackage[centertags]{amsmath}
\usepackage{amssymb}
\usepackage[pdftex]{graphicx}
\usepackage{amsthm}
\usepackage{enumitem}
\usepackage{bbm}
\usepackage{amsfonts}
\usepackage{datetime}

\newcommand{\E}{\mathbb E}
\newcommand{\R}{\mathbb R}
\newcommand{\prob}{\mathbb P}

\newcommand{\indicator}[1]{\mathbbm{1}_{ {#1}  }}

\newcommand{\sfrac}[2] {\mbox{$\frac{#1}{#2}$}}

\newtheorem{teo}{Theorem}[section]

\newtheorem{cor}[teo]{Corollary}
\newtheorem{prop}[teo]{Proposition}
\newtheorem{lem}[teo]{Lemma}

\newdateformat{ukvardate}{\THEDAY\ \monthname[\THEMONTH]\ \THEYEAR}

\title[Large deviations via matrix products]
{The semi-infinite asymmetric exclusion process:\\ large deviations via matrix products \vspace{-0.3cm}
}

\author[Horacio G. Duhart, Peter M\"orters and Johannes Zimmer]{Horacio Gonz\'alez Duhart, Peter M\"orters and Johannes Zimmer}

\begin{document}
%%%%%%%%%%%%%%%%%%%%%%%%%%%%%%%%%%%%%%%%%%
%%%%%%%%%%%%%%%%%%%%%%%%%%%%%%%%%%%%%%%%%%
%%%%%%%%%%%%%%%%%%%%%%%%%%%%%%%%%%%%%%%%%%
\centerline{\rm \small}

\begin{abstract}
We study the totally asymmetric exclusion process on the positive integers with a single particle source at the origin. Liggett (1975) has shown that the long term behaviour of this process has a phase transition: If the particle production rate at the source and the original density are below a critical value, the stationary measure is a product measure, otherwise the stationary measure is spatially correlated. Following the approach of Derrida \emph{et al.}~(1993) it was shown by Grosskinsky~(2004) that these correlations can be described by means of a matrix product representation. In this paper we derive a large deviation principle with explicit rate function for the particle density in a macroscopic box based on this representation. The novel and rigorous technique we develop for this problem combines spectral theoretical and combinatorial ideas and is potentially applicable to other models described by matrix products.
\end{abstract}

\maketitle

\vspace{-.5cm}

{\footnotesize {\bf MSC classification:} Primary 60F10; Secondary 37K05, 60K35, 82C22.\\[-1mm]
boundary, phase transition, out of equilibrium.}

%\vspace{-.2cm}
  
%Introduction
%%%%%%%%%%%%%%%%%%%%%%%%%%%%%%%%%%%%%%%%%%
%%%%%%%%%%%%%%%%%%%%%%%%%%%%%%%%%%%%%%%%%%
%%%%%%%%%%%%%%%%%%%%%%%%%%%%%%%%%%%%%%%%%%
\section{Introduction}
Many natural systems are not in thermodynamic equilibrium, which loosely speaking means that there is a permanent exchange of energy or matter of the system with its surroundings or within the system itself. In statistical physics, the asymmetric exclusion process is often considered the paradigm of such a system out of equilibrium. In the absence of a general theory for systems out of equilibrium, it has been argued that large deviation rate functions play an important role as a replacement for the thermodynamical potential~\cite{BSG}. The principal aim of the present paper is to develop a rigorous mathematical technique to derive such rate functions from a particular type of representation of the stationary state of the system, the matrix products, which twenty years after the pioneering work of Derrida \emph{et al.}~\cite{DEHP} is available for a wide range of particle systems 
out of equilibrium, see for example Blythe and Evans~\cite{BE} for a survey.
\medskip

We present our method in the case of the totally asymmetric exclusion process (TASEP) on the positive integers with a single particle source at the origin, a case which has apparently not been treated in the literature so far. In this Markovian model, particles are positioned on the sites of the semi-infinite lattice $\mathbb N = \{1, 2, \dots \}$ in such a way that no site carries more than one particle. The dynamics of the model can be informally described as follows: A particle source carries a Poisson clock with intensity $\alpha>0$. If this clock rings, the source attempts to inject a particle at site one. If this site is vacant the injection takes place, otherwise it is suppressed and nothing happens. Also, every particle in the system carries an independent Poisson clock with rate one, and when the clock rings the particle tries to jump to the neighbouring site on its right. If this site is vacant the jump takes place, otherwise it 
is suppressed. Note that the exclusion interaction originating from the suppression of jumps and injections ensures that no site ever carries more than one particle.
\pagebreak[3]
\medskip

The exclusion interaction in this model has a profound effect on the behaviour of the system. Most notably the detailed balance
equations for this Markov chain have no nontrivial solution. Hence the system is not reversible, in other words it is out of
equilibrium. The long-term behaviour of the process shows local convergence to a stationary measure which depends on the initial
configuration of the system. Assuming that initially particles are iid Bernoulli with density~$\rho$, this stationary measure has
an interesting phase transition described by Liggett~\cite{liggettergodic}. If the injection rate $\alpha$ satisfies
$\alpha\leq\tfrac12$ and $\rho\leq 1-\alpha$ the system does not feel the interaction and the stationary measure is the product
measure with density $\alpha$. If however $\alpha>\tfrac12$ or $\rho>1-\alpha$, the exclusion of particles leads to spatial
correlations in the stationary measure, which is no longer a product measure. In this case, the overall particle density at
stationarity is the maximum of $1/2$ and the initial density $\rho$, independently of the injection rate $\alpha$.  \medskip

There have been considerable efforts to describe the long range correlations of the stationary measures and the microscopic transition kernels in the exclusion process explicitly. For instance, Sasamoto and Williams~\cite{SW} and  Tracy and Widom~\cite{TW} derive explicit formulas from combinatorial identities, and Sasamoto~\cite{S} uses an ansatz based on orthogonal polynomials. A particularly successful approach to describe spatial correlations  is the matrix product ansatz first suggested in 1993~by Derrida, Evans, Hakim and Pasquier~\cite{DEHP} and refined and extended in a large number of papers, see \cite{D,EFM, H} for a few further examples. 
%\medskip
%
Large deviation principles have been derived for the hydrodynamical limits of a range of boundary driven exclusion processes by Bertini and coauthors~\cite{BSGLL, BLM}
and the method should be extendable to our case. In principle, large deviation principles for the  particle density in a macroscopic box then follow from these results 
by contraction, see \cite{BG}. However, the optimisation in path space, which is  required to get an explicit rate function, is often unwieldy and technical as Bahadoran's  
paper \cite{B} readily testifies.%  
\medskip%

In the light of these difficulties it is a natural idea to try and derive large deviation principles directly from the matrix product ansatz. This plan
was carried out by Derrida \emph{et al.}~\cite{DLS} in the case of an asymmetric exclusion process on a finite interval of sites. Key to their method is a saddle point argument, 
which allows to derive an additivity formula  which compares the stationary measure on the interval with stationary measures on complementary subintervals.  
From this formula an explicit rate function for the particle  density is derived. The paper \cite{DLS} was a spectacular  success, but we have not been able to implement this method in the case of a semi-infinite lattice. In a different development, Angeletti \emph{et al.}~\cite{ATBA} show that 
already for matrix product representations with  finite matrices the large deviation principles that arise from this exhibit a rich phenomenology. Finite matrix representations have 
the advantage that they can be studied using the Perron-Frobenius theory, which is unavailable for infinite matrices. Physical examples, however, are almost always based on representations 
by infinite matrices. 
\bigskip

\pagebreak[3]

In this paper we present a rigorous and novel approach to calculate large deviations for the macroscopic particle density in the
semi-infinite totally asymmetric exclusion process. We use the matrix product representation as a starting point, and base the
analysis on the G\"artner-Ellis theorem. To study the asymptotics of the cumulant generating function of the particle density, we
use quite different approaches for the lower and upper bounds. The lower bound is based on the spectral theory of Toeplitz
operators in a suitable weighted sequence space, while the upper bound exploits combinatorial identities coming directly from the
matrix product ansatz. As our method is not too technical, we believe that it is very promising to deal with a wide range of other
particle systems whose stationary measure can be described by a matrix product representation.%
\bigskip%

The paper is organised as follows. In Section~2 we give a rigorous definition of the~model, state background results and formulate 
and interpret our main result. Section 3 discusses the matrix product representation in this case and describes our approach to the large deviation problem. 
The proof of the upper bound for the  cumulant generating function is carried out in Section 4, while the lower bound is derived in Section 5. The proof is
completed in Section 6, in which we also provide some further comments on our technique.
%\bigskip

% The TASEP
%%%%%%%%%%%%%%%%%%%%%%%%%%%%%%%%%%%%%%%%%%
%%%%%%%%%%%%%%%%%%%%%%%%%%%%%%%%%%%%%%%%%%
%%%%%%%%%%%%%%%%%%%%%%%%%%%%%%%%%%%%%%%%%%
\section{The semi-infinite TASEP}
\subsection{Background}

To give a formal definition of the model, we first define the auxiliary switching and swapping functions \mbox{$\sigma^{x},\sigma^{x,y}\colon\{0,1\}^{\mathbb N}\to\{0,1\}^{\mathbb N}$} by
\[(\sigma^{x,y}\eta)_z=\begin{cases} \eta_y & \text{if }z=x \\ \eta_x & \text{if }z=y \\ \eta_z & \text{if }z\notin \{x,y\}, \\ \end{cases} \qquad
(\sigma^{x}\eta)_z=\begin{cases} 1-\eta_x & \text{if }z=x \\ \eta_z & \text{if }z\neq x.  \end{cases} \]
A semi-infinite totally asymmetric exclusion process (TASEP) with injection rate $\alpha\in(0,1)$ is a Markov process $\{\xi(t)\}_{t\geq0}$ in continuous time with state space $\{0,1\}^{\mathbb N}$ and semigroup $S(t)$ identified by its infinitesimal generator $G$ defined by
\begin{align}\label{generator}
 Gf(\eta) & =  \alpha(1-\eta_1)\left(f(\sigma^{1}\eta)-f(\eta)\right) \nonumber\\
          & +  \sum_{k\in\mathbb N} \eta_k(1-\eta_{k+1})\left(f(\sigma^{k,k+1}\eta)-f(\eta)\right),
\end{align}
where $f\colon\{0,1\}^{\mathbb N}\to\mathbb R$ is a function that depends only on a finite number of sites.
\medskip

Denote by $\nu_\alpha$ the \emph{product measure} with constant density $\alpha$, that is
\[\nu_\alpha\{\eta\in\{0,1\}^{\mathbb N} \colon \eta_{j_1}=1,\eta_{j_2}=1,
\ldots,\eta_{j_n}=1\}=\alpha^n\]for all distinct choices of
$j_1,j_2,\ldots,j_n\in\mathbb N$ and all $n\in\mathbb N$.
We say a measure $\mu_\rho$ on $\{0,1\}^{\mathbb N}$ is \emph{asymptotically product} with density $\rho$ if
\[\lim_{k\rightarrow\infty}\mu_\rho\{\eta\in\{0,1\}^{\mathbb N}\colon \eta_{j_1+k}=1,\eta_{j_2+k}=1,
\ldots,\eta_{j_n+k}=1\}=\rho^n\]for all distinct choices of
$j_1,j_2,\ldots,j_n\in\mathbb N$ and all $n\in\mathbb N$.
It is known that the semi-infinite TASEP is not an ergodic process and there is no uniqueness of the stationary
measure as proved in Theorem 1.8 of~\cite{liggettergodic}.
%\pagebreak[3]

\begin{teo}\cite[Theorem 1.8]{liggettergodic}\label{liggettteo}
Let $\mu$ be a product measure on $\{0,1\}^{\mathbb N}$ for which
$\rho:=\lim_{k\rightarrow\infty}\mu\{\eta:\eta_k=1\}$ exists. Then
there exist probability measures $\mu_\varrho^\alpha$ defined if
\begin{itemize}
\item either $\alpha\leq \frac12$ and $\varrho>1-\alpha$, 
\item or $\alpha> \frac12$ and $\frac12 \leq \varrho \leq 1$,
\end{itemize}
which are asymptotically product with density $\varrho$, such that
%\[\begin{array}{rcc}
%   \text{if } \alpha\leq\frac{1}{2}\text{ then } & {\displaystyle\lim_{t\rightarrow\infty}\mu
%S(t)}= & \left\{\begin{array}{cl}  \nu_\alpha & \text{ if }\rho\leq1-\alpha \\[2mm]
%\mu_\rho^\alpha & \text{ if }\rho>1-\alpha, \\ \end{array} \right.\\
% & & \\
%   \text{and if } \alpha >  \frac{1}{2}\text{ then } & {\displaystyle\lim_{t\rightarrow\infty}\mu
%S(t)}= & \left\{\begin{array}{cl} \mu_{1/2}^\alpha & \text{if
% $\rho\leq\frac{1}{2}$}\phantom{-\alpha}\\ & \\ \mu_\rho^\alpha & \text{if
% $\rho>\sfrac{1}{2}$.}\phantom{-\alpha}  \\ \end{array}  \right.\\
%  \end{array}\]
  
\[\begin{array}{rl}
   \text{if } \alpha\leq\frac{1}{2}\text{ then } & {\displaystyle\lim_{t\rightarrow\infty}\mu
S(t)}=\begin{cases}  \nu_\alpha & \text{ if }\rho\leq1-\alpha \\
\mu_\rho^\alpha & \text{ if }\rho>1-\alpha, \end{cases} \\[4mm]
   \text{and if } \alpha >  \frac{1}{2}\text{ then } & {\displaystyle\lim_{t\rightarrow\infty}\mu
S(t)}= \begin{cases} \mu_{1/2}^\alpha & \text{if }\rho\leq\frac{1}{2}\phantom{-\alpha,}\\ 
  \mu_\rho^\alpha & \text{if }\rho>\sfrac{1}{2}.\phantom{-\alpha} \end{cases} \\
  \end{array}\]  
  
\end{teo}

Observe the notation we adopt throughout this paper: The \emph{initial} particle density for the TASEP is denoted~$\rho$, whereas
$\varrho$ is the \emph{stationary} particle density under~$\mu^\alpha_\varrho$.  Convergence in Theorem~\ref{liggettteo} is not
uniform and slows down the further sites are from the origin.  The initial density~$\rho$ determines the stationary measure the
process will converge to. In particular, starting from all sites empty, if the injection rate satisfies
~\smash{$\alpha\leq \frac12$}, the distribution of the process converges to the product measure with constant~density~$\alpha$. If
\smash{$\alpha>\frac12$}, the distribution of the process converges to~\smash{$\mu_{^{\frac12}}^\alpha$}, which has spatial
correlations and an asymptotic density equal to~\smash{$\frac12$}. Observe that the injection mechanism is not able to produce a
stationary particle density larger than~\smash{$\frac12$} unless there is initially a high density of particles in the system.
\medskip

We  will explore matrix representations of the measures  $\mu_\varrho^\alpha$ in Section~3.

\subsection{Main result}

Our problem at hand is to find the rate function for a large deviation principle of the empirical density on the first $n$ sites
\begin{equation}\label{empdensity}Z_n=\frac{1}{n}\sum_{k=1}^n\eta_k\end{equation}
under the stationary measure, as $n$ goes to infinity.  Recall from Theorem~2.1 that the answer to this problem depends in a
subtle way on the spatial correlations occurring in the case $\alpha>\frac12$.  \medskip

The theory of large deviations analyses the exponential decay of probabilities of increasingly unlikely events. Formally, a sequence of random variables $\{Z_n\}_{n\in\mathbb N}$ 
taking values in $[0,1]$ satisfies a large deviation principle with rate function $I\colon[0,1]\to[0,\infty]$ under the probability measure $\prob$ if
\begin{enumerate}
\item[(i)] the function $I$ is lower semicontinuous,
\item[(ii)] for all open sets $G\subset [0,1]$ we have \[\liminf_{n\to\infty}\frac{1}{n}\log\prob\{Z_n\in G\}\geq -\inf_{z\in G}I(z),\]
\item[(iii)] and for all closed sets $F\subset [0,1]$ we have \[\limsup_{n\to\infty}\frac{1}{n}\log\prob\{Z_n\in F\}\leq -\inf_{z\in F}I(z).\]
\end{enumerate}

The main result of this paper is the following large deviation principle.

\pagebreak[3]

\begin{teo}\label{main}
Let $\{Z_n\}_{n\in\mathbb N}$ be the sequence of random variables defined as the empirical density (\ref{empdensity}) of a semi-infinite TASEP with injection rate $\alpha\in(0,1)$ and initial asymptotically product measure for which $\rho$ as defined in Theorem \ref{liggettteo} exists. Then, under the stationary probability measure given by Theorem~\ref{liggettteo}, $\{Z_n\}_{n\in\mathbb N}$ satisfies a large deviation principle with convex rate function $I\colon [0,1]\to[0,\infty]$ given as follows. 
\begin{enumerate}
\item[(a)] If $\alpha\leq\frac{1}{2}$ and $\rho<1-\alpha$,  then \[I(z)= z\log \frac{z}{\alpha} + (1-z)\log \frac{1-z}{1-\alpha}.\]  
%\\[-2mm]
%\end{enumerate}
%\begin{enumerate}
\item[(b)] If $\alpha>\frac{1}{2}$  and $0\leq \rho\leq \frac{1}{2}$, then 
\[I(z)=\left\{\begin{array}{ccc} {\displaystyle
 z\log\frac{ z}{\alpha} + (1-z)\log\frac{1-z}{1-\alpha}+\log\left(4\alpha(1-\alpha)\right)} & \text{ if }&0 \leq  z \leq 1-\alpha, \\[2mm] 
{\displaystyle 2\left[z\log z+(1-z)\log(1-z)+\log2 \right]} & \text{ if }& 1-\alpha < z \leq \frac{1}{2}, \\[2mm]
{\displaystyle z\log z+(1-z)\log(1-z)+\log2}  & \text{ if }&\frac{1}{2}  <  z \leq  1.
\end{array}\right.\] \\[-2mm]
\item[(c)] If $\alpha>\frac{1}{2}$  and $\frac{1}{2}< \rho< \alpha$, then 
\[I(z)=\left\{\begin{array}{ccc} {\displaystyle
 z\log\frac{ z}{\alpha} + (1-z)\log\frac{1-z}{1-\alpha}+\log\frac{\alpha(1-\alpha)}{\rho(1-\rho)}} & \text{ if }&0 \leq  z \leq 1-\alpha, \\[2mm] 
{\displaystyle 2\left[z\log z+(1-z)\log(1-z)-\log\sqrt{\rho(1-\rho)} \right]} & \text{ if }& 1-\alpha < z \leq 1-\rho, \\[2mm]
{\displaystyle z\log \frac{z}{\rho}+(1-z)\log\frac{1-z}{1-\rho}}  & \text{ if }& 1-\rho <  z \leq  1.
\end{array}\right.\]
\end{enumerate}
\end{teo}
\medskip

Recall that part~$(a)$ is well-known and included for completeness. It implies the weak law of large numbers, saying that the
empirical density converges in probability to $\alpha$ if $\alpha\leq\tfrac{1}{2}$, see Figure~\ref{ratesa}. The unique zero of
the rate function moves from $0$ to $\tfrac{1}{2}$ with the value of $\alpha$, at the same time it is getting easier to achieve
any given density larger~than~\smash{$\frac12$}.  \medskip

\begin{figure}[!htb]
\begin{center}
\vskip-0.25in
\includegraphics[width=0.4\textwidth]{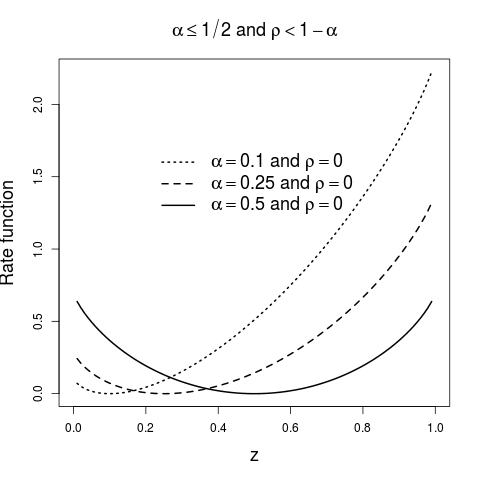}
\includegraphics[width=0.4\textwidth]{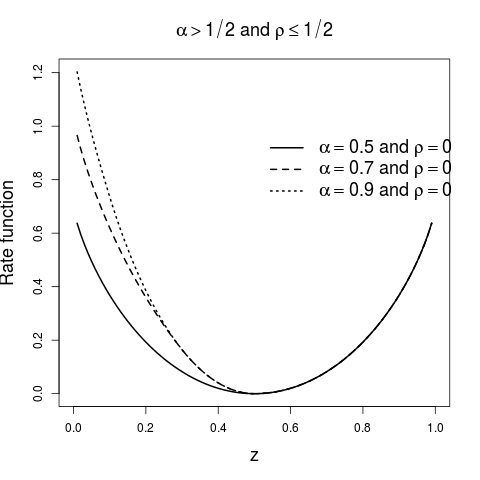}
\end{center}
\caption{Rate functions for case~\emph{(a)} and~\emph{(b)} of Theorem \ref{main}.}\label{ratesa}\label{rate2}
\end{figure} 

Part~$(b)$ shows that the empirical density converges to \smash{$\frac{1}{2}$} if \smash{$\alpha>\frac{1}{2}$} and $\rho<\tfrac{1}{2}$. 
Looking at Figure~\ref{rate2} we see that the rate function is still convex and its zero stays fixed at \smash{$\frac{1}{2}$}. Now reaching high densities has always 
the same cost regardless of the value of $\alpha$, but low densities become increasingly expensive as the value of $\alpha$ increases. 
Note that the rate function is non-analytic at the value $z=1-\alpha$, which reveals a
dynamical phase transition in the sense of~\cite{NT16}. 
\pagebreak[3]
%\medskip

\begin{figure}[!htb]
\begin{center}
\includegraphics[width=0.4\textwidth]{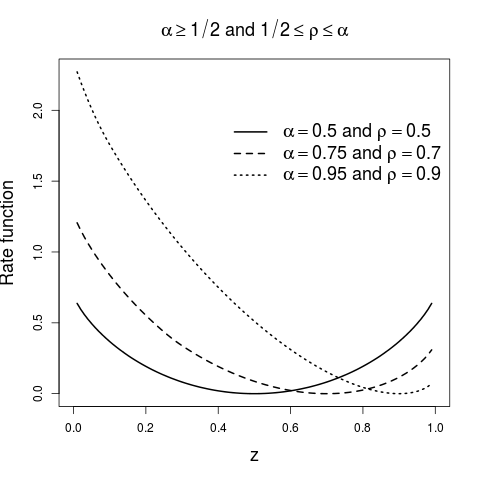}
\includegraphics[width=0.4\textwidth]{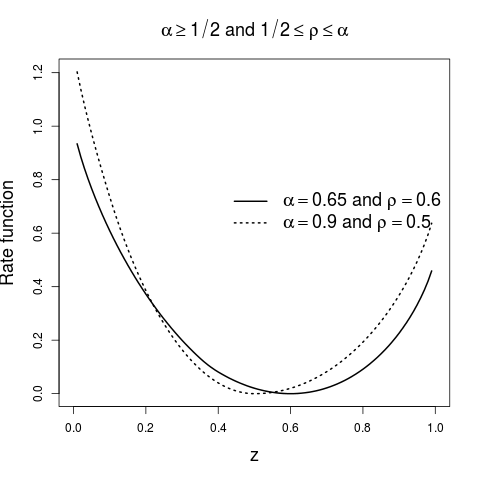}
\vspace{-5mm}
\end{center}
\caption{Rate functions for case \emph{(c)} of Theorem \ref{main}.}\label{rate3}\label{morerate3}
\end{figure} 

For part (c), the minimum of the rate function is now at $\rho$. Low densities still become increasingly expensive as $\alpha$
increases; yet, high densities now become cheaper, see the left diagram of Figure~\ref{rate3}. In this case, we observe that both
$\alpha$ and $\rho$ play a role in the rate function, see the right diagram of Figure~\ref{morerate3} to appreciate the joint
effect. The phase transitions are seen at $z=1-\alpha$, as in the previous case, and $z=1-\rho$.  As $\rho\to \alpha$ we recover
the rate function of Bernoulli product measures.  \medskip

\begin{figure}[!h]
\begin{center}
\includegraphics[width=0.4\textwidth]{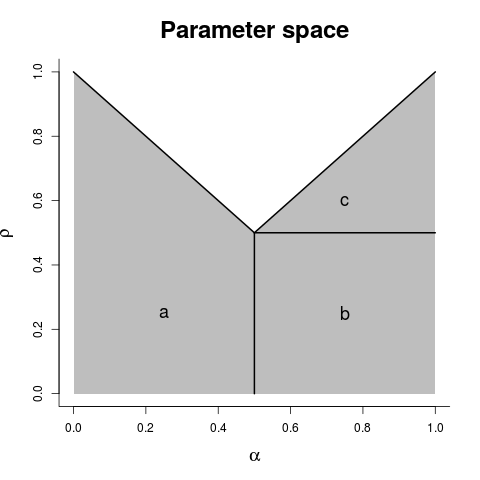}
\end{center}
\vspace{-5mm}
\caption{The range of validity of Theorem~\ref{main} in parameter space is shaded in grey.}\label{pspace}
\end{figure} 

In parts~(b) and~(c), the cost in the first regime, when $z\leq 1-\alpha$,  is up to a shift by a negative constant equal to the cost of changing the boundary density. The rates in the third regime, when $1-z$ is smaller than the typical density, represent the cost of replacing the typical density in the asymptotic Bernoulli product measure by the desired value, so the cost in this regime is bulk dominated. In the second regime the cost is larger than in the third, indicating that both a bulk and a boundary cost have to be paid.
\medskip

The regimes when $\alpha\leq\frac12$ and $\rho\geq 1-\alpha$, or when
$\alpha>\frac12$ and $\rho\geq \alpha$ are not covered by 
% P replaced our theorem  by
the techniques of this paper  and remain open, see Figure \ref{pspace}.
% P added
In the former case we have no matrix product representation.
\medskip

If $\alpha>\frac12$ and $\rho<\alpha$ it is worth comparing our large deviation result for the semi-infinite TASEP with that for the finite TASEP studied in \cite{DLS}. The stationary measure on the semi-infinite TASEP can be obtained as a limit of the stationary measures on the finite TASEP with $n$ sites and boundary densities chosen as $\alpha$ on the left boundary, and $\min\{1-\rho,\tfrac{1}{2}\}$ on the right boundary, see~\cite[Section~3]{liggettergodic}. 
However, the transition rates across bonds in the semi-infinite TASEP are not equal to $\min\{1-\rho,\tfrac{1}{2}\}$. It is therefore somewhat surprising that  the large deviation rate of the sequence of densities in the finite TASEP with increasing system size obtained in~\cite[(3.12)]{DLS} still agrees with the rate we have obtained for the
average density over increasing blocks in the  infinite TASEP at stationarity.  
\medskip
\pagebreak[3]

%MPA
%%%%%%%%%%%%%%%%%%%%%%%%%%%%%%%%%%%%%%%%%%
%%%%%%%%%%%%%%%%%%%%%%%%%%%%%%%%%%%%%%%%%%
%%%%%%%%%%%%%%%%%%%%%%%%%%%%%%%%%%%%%%%%%%
\section{Matrix product ansatz}

Grosskinsky~\cite{grosskinsky}, following the seminal work of~\cite{DEHP}, has given a characterisation 
of the long range dependence in $\mu^\alpha_{\varrho}$ with a matrix product ansatz.

\begin{teo}\cite[Theorem 3.2]{grosskinsky}\label{grossteo}
Suppose there exist (possibly infinite) nonnegative matrices $D$, $E$ and vectors $w$ and $v$, fulfilling the algebraic
relations
\begin{subequations}
\label{mpaconditions} 
\begin{align}
   DE&=D+E, \label{mpaconditionsa} \\
\alpha w^{\sf T}E&=w^{\sf T}, \label{mpaconditionsb}\\
c(D+E)v&=v, \label{mpaconditionsc}
  \end{align}
\end{subequations}
for some $c>0$. Then
\begin{itemize}
\item[(a)] the probability measure $\bar\nu_c^\alpha$ defined by
\begin{equation}\bar\nu^\alpha_c\{\zeta\in\{0,1\}^{\mathbb N} \colon  \zeta_1=\eta_1,\ldots,\zeta_n=\eta_n\}
 =\frac{w^{\sf T}\left(\prod_{k=1}^n\eta_kD+(1-\eta_k)E \right)v}{w^{\sf T}(D+E)^nv} \label{grossmeasure}\end{equation}
is invariant for the generator~\eqref{generator} if and only if 
\begin{itemize}
\item[$\bullet$] either $\alpha\leq\tfrac{1}{2}$ and $0\leq c\leq\alpha(1-\alpha)$ 
\item[$\bullet$] or $\alpha>\tfrac{1}{2}$ and  $0\leq c\leq\frac{1}{4}$.\\[-1mm]
\end{itemize}
\item[(b)] The measure $\bar\nu^\alpha_c$  has stationary current $\E_{\bar\nu^\alpha_c}[\eta_k(1-\eta_{k+1})]=c$, for all $k\ge 1$.
It equals $\nu_\alpha$ if $c=\alpha(1-\alpha)$ and $\alpha\leq\tfrac{1}{2}$, and otherwise it is asymptotically product with density $\varrho$ given 
as the solution of  $c=\varrho(1-\varrho)$ which satisfies $\varrho\geq \frac12$.
\end{itemize}
\end{teo}
\medskip

\pagebreak[3]

Very often, to apply Theorem~\ref{grossteo}, no explicit solution of~\eqref{mpaconditions} is needed. Below we only use the recursive structure of these equations to show that the measures \smash{$\bar\nu_{c}^\alpha$} and \smash{$\mu_{\varrho}^\alpha$} agree under certain conditions. 
Note that this does not follow directly from Theorem~\ref{liggettteo} 
as this result does not describe the long-term behaviour of the TASEP started in 
\smash{$\bar\nu_{c}^\alpha$}, which is not necessarily a product measure.
\smallskip

\begin{prop}\label{unique} If $\alpha\geq \frac12$, $\varrho\geq \frac12$ and $c=\varrho(1-\varrho)$, then
the measures $\bar\nu_c^\alpha$ and $\mu_\varrho^\alpha$ agree.\end{prop}

\begin{proof}
By part (e) in \cite[Theorem 3.10]{liggettergodic} the measure $\mu_\varrho^\alpha$ is uniquely 
determined by the following two properties, numbered as in \cite{liggettergodic},
\begin{itemize}
\item[(c)] If $u,n\in\mathbb N$ with $1<u<u+1<n$, and $\eta\in\{0,1\}^{n}$ with $\eta_u=1$, $\eta_{u+1}=0$, then
\begin{align*}
\mu_\varrho^\alpha & \{\zeta:\zeta_k  =\eta_k\text{ for }k\leq n\}\\ & =
c\, \mu_\varrho^\alpha\{\zeta:\zeta_k=\eta_k\text{ for } k\leq u-1, \zeta_k=\eta_{k+1}\text{ for }u+1\leq k\leq n-1\}.
\end{align*}
\item[(d)] If $n>1$ and $\eta\in\{0,1\}^{n}$ with $\eta_1=0$, then
\begin{align*}
\alpha\mu_\varrho^\alpha \{\zeta:\zeta_k  =\eta_k\text{ for }k\leq n\} =c\,\mu_\varrho^\alpha\{\zeta:\zeta_k=\eta_{k+1}\text{ for } k\leq n-1\}.
\end{align*}
\end{itemize}
We show that $\bar\nu_c^\alpha$ satisfies these properties. Under the assumptions of (c) we get from properties (\ref{mpaconditionsa}) in the second equality and (\ref{mpaconditionsc}) in the third one
\begin{align*}
\bar\nu_c^\alpha \{\zeta:\zeta_k & =\eta_k\text{ for }k\leq n\}\\ &=\frac{w^{\sf T}\big(\prod_{k=1}^{u-1}\eta_kD+(1-\eta_k)E\big)DE\big(\prod_{k=u+2}^{n}\eta_kD
+(1-\eta_k)E\big)v}{w^{\sf T}(D+E)^nv}\\
& =\frac{w^{\sf T}\big(\prod_{k=1}^{u-1}\eta_kD+(1-\eta_k)E\big)(D+E)\big(\prod_{k=u+2}^{n}\eta_kD+(1-\eta_k)E\big)v}{w^{\sf T}(D+E)^nv}\\
&=c\,\bar\nu_c^\alpha\{\zeta:\zeta_k=\eta_k\text{ for } k\leq u-1, \zeta_k=\eta_{k+1}\text{ for }u+1\leq k\leq n-1\}.
\end{align*}
Under the assumptions of (d) we get from conditions (\ref{mpaconditionsb}) in the second equality and (\ref{mpaconditionsc}) in the third one,
\begin{align*}
\alpha\bar\nu_c^\alpha \{\zeta:\zeta_k  =\eta_k\text{ for }k\leq n\}
&=\alpha \frac{w^{\sf T}E\left(\prod_{k=2}^{n}\eta_kD+(1-\eta_k)E\right)v}{w^{\sf T}(D+E)^nv}\\
& =\frac{w^{\sf T}\left(\prod_{k=2}^{n}\eta_kD+(1-\eta_k)E\right)v}{w^{\sf T}(D+E)^nv}\\
&=c\,\bar\nu_c^\alpha\{\zeta:\zeta_k=\eta_{k+1}\text{ for } k\leq n-1\}.
\end{align*}
Hence $\bar\nu_c^\alpha$ satisfies (c) and (d) and therefore agrees with $\mu_\varrho^\alpha$.
\end{proof}
\medskip

We now explain our approach to find a large deviation principle of the empirical density under this measure. 
We will approach this via the G\"artner-Ellis theorem, see Theorem~V.6 in \cite{hollander}. Here we state 
the conditions specific for our case.

\begin{teo}(G\"artner-Ellis)\label{GE}
Let $\{Z_n\}_{n\in\mathbb N}$ be a sequence of random variables on a probability space $(\Omega,\mathcal A,\prob)$, where $\Omega$ is a non-empty subset of $\mathbb R$. If the limit cumulant generating function $\Lambda\colon\mathbb R\to\mathbb R$ defined by 
\[\Lambda(\theta)=\lim_{n\rightarrow\infty}\tfrac{1}{n}\log\E[e^{n\theta Z_n}]\]
exists and is differentiable on all $\mathbb R$, then $\{Z_n\}_{n\in\mathbb N}$ satisfies a large deviation principle with rate function $I\colon\Omega\to[-\infty,\infty]$ defined by
\[I(z)=\sup_{\theta\in\mathbb R}\{z\theta-\Lambda(\theta)\}.\]
\end{teo}
\medskip

To calculate the moment generating function $M_n(\theta)$ of $Z_n$ we use Theorem~\ref{grossteo} and Proposition \ref{unique} in the third equality and condition~\eqref{mpaconditionsc} in the fifth one, to get
\begin{align*}
   M_n(\theta) & =  \E\big[e^{n\theta Z_n}\big] =  \E\Big[\exp\big(\theta\sum_{k=1}^n\xi_k\big)\Big]\\
   &  = \sum_{\eta\in\{0,1\}^n}\bar\nu_{c}^\alpha\{\xi:\xi_k=\eta_k\text{ for }k\leq n\}\exp\left(\theta\sum_{k=1}^n\eta_k\right)\\
   & = \frac{w^{\sf T}(e^\theta D+E)^nv}{ w^{\sf T}(D+E)^nv} = \frac{c^n}{w^{\sf T}v} w^{\sf T}(e^\theta D+E)^nv,
  \end{align*}
and the cumulant generating function simplifies to
\begin{align} 
\Lambda(\theta) &  =  \lim_{n\rightarrow\infty}\tfrac{1}{n}\log M_n(\theta) \notag\\
& = \lim_{n\rightarrow\infty}\tfrac{1}{n}\log w^{\sf T}(e^\theta D+E)^nv +\log c.
 \label{shortlambda}\end{align}
If $D$ and $E$ were finite matrices, we could identify this limit using the Perron-Frobenius theorem as the spectral radius of the matrix
$e^\theta D+E$. However in our example (and in almost all physically interesting examples) the matrices solving~\eqref{mpaconditions} are 
necessarily infinite. A first idea would be to truncate the matrices to finite size, calculate the spectral radius and taking a limit, 
but this turns out to lead to a wrong result, as it neglects the important information contained in the vectors $v$ and $w$.
\medskip

\pagebreak[3]

We will look at lower and upper bounds in~\eqref{shortlambda} separately. For the \emph{upper bound} we exploit that matrices $D$ and $E$, as well as the vectors $v$ and $w$, solving~\eqref{mpaconditions} are explicitly known. We introduce weighted $\ell^2$ spaces, denoted $\ell^2_s$, and interpret the matrix $e^\theta D+E$ as an operator on these spaces. If the weights are such that $v$ is an element of $\ell^2_s$, and $w^{\sf T}$ an element of its dual, we can get a bound on~\eqref{shortlambda} from the spectral radius of the operator, which can be optimised by minimising the bound over all admissible weights. In order to obtain the spectral radius we use a simple isomorphism between weighted and unweighted $\ell^2$ spaces. Acting on the unweighted spaces, the operator has a Toeplitz structure and from the general theory of Toeplitz operators on $\ell^2$ an explicit formula for the spectral radius is available. This estimate will be carried out in detail in Section~4. \medskip

The technique for the \emph{lower bound} only relies on the structure of the equations~\eqref{mpaconditions}. These provide an algorithmic way to reduce arbitrary products of $D$ and $E$ to linear combinations of monomials of the form $E^{p-j}D^j$ with $0\leq j\leq p$. We expand the product $(e^\theta D+E)^n$ into a polynomial with nonnegative coefficients $f^n_{p,j}(\theta)$, and use~\eqref{mpaconditions} to derive a recursion formula for the coefficients. While it seems to be too complicated to fully resolve this recursion, we focus on selected key terms which can be derived explicitly. Note that for a lower bound we can drop all inaccessible terms in the expansion. In our case, to obtain the growth rate it suffices to include the fastest growing terms \smash{$f^n_{p,p} D^p$} and \smash{$f^n_{p,0} E^p$} corresponding to all summands in the product which can be reduced to monomials of just one variable. It is quite a typical phenomenon that only a small number of boundary summands contribute to the growth of the sum, and that these coefficients can be identified without solving the entire system of equations.  This estimate will be carried out explicitly in Section~5.  \medskip

%Upper bound
%%%%%%%%%%%%%%%%%%%%%%%%%%%%%%%%%%%%%%%%%%
%%%%%%%%%%%%%%%%%%%%%%%%%%%%%%%%%%%%%%%%%%
%%%%%%%%%%%%%%%%%%%%%%%%%%%%%%%%%%%%%%%%%%
\section{Upper bound: Spectral theory of Toeplitz operators}
\subsection{Weighted $l^2$ spaces}

To find an upper bound for the cumulant generating function we consider the weighted spaces \[\ell^2_s=\{x=(x_k)_{k\in\mathbb N}\colon \sum_{k=1}^\infty |x_k|^2s^k<\infty\} \] with $s>0$. Note that imposing $s=1$ recovers the usual $\ell^2$ space. Moreover for fixed $s$, $\ell^2_s$ with its corresponding norm \[|x|_{\ell^2_s}^2=\sum_{k=1}^\infty|x_k|^2s^k\] is a Banach space.
The next lemma will help us to translate classic $\ell^2$ theory to $\ell^2_s$.

\begin{lem}\label{transform} The function $T_s\colon\ell^2\rightarrow\ell^2_s$ defined by \[(T_sx)_k=\frac{x_k}{s^{k/2}}\] for $s>0$ is a bijective isometry.\end{lem} 
\begin{proof}
 We can define the inverse $T_s^{-1}\colon\ell^2_s\rightarrow\ell^2$ by $(T_s^{-1}x)_k=x_ks^{k/2}$ 
 and hence $T_s$ is bijective. We just need to prove it is an isometry, so let $x\in\ell^2_s$ and calculate 
\[|T_s^{-1}x|^2_{\ell^2}  =  \sum_{k=1}^\infty |(T_s^{-1}x)_k|^2 = 
\sum_{k=1}^\infty|x_ks^{k/2}|^2  = |x|_{\ell^2_s}^2.\]
Analogously, for $x\in\ell^2$ we have $|T_sx|_{\ell^2_s}=|x|_{\ell^2}$.
\end{proof}
\pagebreak[3]

\begin{lem} The dual space $\ell^{2\ast}_s$ can be identified with $\ell^2_{s^{-1}}$. \end{lem}

\begin{proof}
Define the dual product $\langle\cdot,\cdot\rangle_D\colon\ell^2_{s^{-1}}\times\ell^2_s \to \R$ by
$\langle y,x\rangle_D =\langle T_{s^{-1}}^{-1} y,T_{s}^{-1}x\rangle$,
where $\langle\cdot,\cdot\rangle$ is the usual inner product in $\ell^2$.
We first prove that for each vector $y\in\ell^2_{s^{-1}}$ there exists a function $f_y\in\ell^{2\ast}_s$ such that $f_y(x)=\langle y,x\rangle_D$.
To this end, let $y\in\ell^2_{s^{-1}}$ and define $f_y\colon\ell^2_s\to\mathbb R$ by \[f_y(x)=\langle y,x\rangle_D=\sum_{k\in\mathbb N}x_ky_k.\]
The linearity of $f_y$ follows easily from the definition; the Cauchy-Schwarz inequality in $\ell^2$ shows it is also bounded,
\begin{align*}
|f_y(x)| & = |\sum_{k\in\mathbb N}x_ky_k|
 = |\sum_{k\in\mathbb N}x_ks^{\tfrac{k}{2}}y_ks^{-\tfrac{k}{2}}|\\
& \leq \left(\sum_{k\in\mathbb N}|x_k|^2s^k \right)^{\tfrac{1}{2}} \left(\sum_{k\in\mathbb N}|y_k|^2s^{-k} \right)^{\tfrac{1}{2}}
 = |x|_{\ell^2_s}  |y|_{\ell^2_{s^{-1}}}.
\end{align*}
Conversely, let $f\in\ell^{2\ast}_s$. Define $g\colon\ell^2\to\mathbb R$ by $g(x)=(f\circ T_s)(x)$. Since $f$ and $T_s$ are both linear, so is $g$, and since $f$ is bounded,
\[|g(x)|=|(f\circ T_s)(x)|\leq |f|_{\ell_s^{2\ast}}|T_s(x)|_{\ell^2_s}=|f|_{\ell_s^{2\ast}}|x|_{\ell^2}<\infty.\]
Hence, $g\in\ell^{2\ast}$ and by the Riesz Representation theorem there exists a unique $\tilde y\in\ell^2$ such that 
$g(x)=\langle x,\tilde y\rangle$ for all $x\in\ell^2$.
Let $y=T_{s^{-1}}\tilde y\in\ell^2_{s^{-1}}$. Since $T_s$ is invertible we have that for all $x\in\ell^2_s$
\begin{align*}
f(x)=(g\circ T^{-1}_s)(x)  = \langle T^{-1}_sx,\tilde y\rangle = \sum_{k\in\mathbb N}(T^{-1}_sx)_k(T_{s^{-1}}^{-1}y)_k = \sum_{k\in\mathbb N}x_ky_k=\langle y,x\rangle_D,
\end{align*}
whence $f\in\ell^{2\ast}_s$ is represented by $y\in\ell^2_{s^{-1}}$.
\end{proof}
\medskip

We now need an explicit solution for~\eqref{mpaconditions}, so we first define the values
\begin{align}\label{eigen}
\lambda_1 & = \frac{1-2c+\sqrt{1-4c}}{2c},\\
\lambda_2 & = \frac{1-2c-\sqrt{1-4c}}{2c}.
\end{align} 
Elementary calculations show that the matrices $D$, $E$ and the vectors 
$v$ and $w$ defined by
\begin{subequations}
\label{explicit}
\begin{equation}
  D=
\begin{pmatrix}
1 & 1 & 0 & 0 & \cdots \\
0 & 1 & 1 & 0 & \cdots \\
0 & 0 & 1 & 1 & \cdots \\
0 & 0 & 0 & 1 & \ddots \\
\vdots & \vdots & \vdots & \vdots & \ddots \\ 
\end{pmatrix},
\qquad
E=
\begin{pmatrix}
1 & 0 & 0 & 0 & \cdots \\
1 & 1 & 0 & 0 & \cdots \\
0 & 1 & 1 & 0 & \cdots \\
0 & 0 & 1 & 1 & \ddots \\
\vdots & \vdots & \vdots & \vdots & \ddots \\ 
\end{pmatrix},
\end{equation}
\begin{equation}
  w^{\sf T}=
\begin{pmatrix}
1,\frac{1}{\alpha}-1,\left(\frac{1}{\alpha}-1\right)^2,\cdots
\end{pmatrix}
\text{ and }v=\frac{1}{\lambda_1-\lambda_2}
 \begin{pmatrix}
   \lambda_1-\lambda_2\\
 \lambda_1^2-\lambda_2^2\\
 \lambda_1^3-\lambda_2^3\\
 \vdots \\
  \end{pmatrix}
\end{equation}
\end{subequations}
satisfy the matrix product conditions~\eqref{mpaconditions}. 
In the boundary case $c=\frac14$ we have
$\lambda_1=1=\lambda_2$ and we take $v^{\sf T}=(1,2,3,\ldots)$.
\medskip

To simplify notation, define for fixed $\theta\in\mathbb R$
the operator $A(\theta)\colon\ell^2_s\to\ell^2_s$ with the infinite matrix representation 
\begin{equation}\label{toe}A(\theta)=e^\theta D+E=
\begin{pmatrix}
1+e^\theta & e^\theta & 0 & 0 & \cdots \\
1 & 1+e^\theta & e^\theta & 0 & \cdots \\
0 & 1 & 1+e^\theta & e^\theta & \cdots \\
0 & 0 & 1 & 1+e^\theta & \ddots \\
\vdots & \vdots & \vdots & \ddots & \ddots \\ 
\end{pmatrix}\end{equation}
and then the $k$-th component of the vector $A(\theta)x$ is
\[(A(\theta)x)_k=
\begin{cases}x_1(1+e^\theta)+x_2e^\theta & \text{if }k=1 \\
x_{k-1}+ x_k(1+e^\theta)+x_{k+1}e^\theta  & \text{if }k>1.  \end{cases}\]
\medskip

\begin{prop}\label{abounded} The operator $A(\theta)\colon\ell^2_s\rightarrow\ell^2_s$ is bounded. \end{prop}
\begin{proof}
 Let $x\in\ell^2_s$. Using Cauchy-Schwarz in $\mathbb R^2$ and $\mathbb R^3$ for each term of~$A(\theta)x$ gives
\begin{align*}
   |A(\theta)x|^2_{\ell^2_s} & =  \sum_{k=1}^n|(A(\theta)x)_k|^2s^k \\
& =  |x_1(1+e^\theta)+x_2e^\theta|^2s + \sum_{k=2}^n|x_{k-1}+ x_k(1+e^\theta)+x_{k+1}e^\theta |^2s^k \\
& \leq  (x_1^2+x_2^2)((1+e^\theta)^2+e^{2\theta}) + \sum_{k=2}^n(x_{k-1}^2+ x_k^2+x_{k+1}^2)(1+(1+e^\theta)^2+e^{2\theta})s^k \\
& \leq  C_s|x|^2_{\ell^2_s},
%& \leq  2(1+e^\theta+e^{2\theta})(1+s+2s^{-1}+s^{-2})|x|^2_{\ell^2_s},
\end{align*}
where $C_s>0$ is a constant independent of $x$ and hence we see that $A(\theta)$ is a bounded linear operator.
\end{proof}
\medskip

\begin{lem}\label{opnorm} Let $L\in\mathcal L(\ell^2_s)$, that is a bounded linear operator from $\ell^2_s$ to itself. 
The operator $\tilde L = T^{-1}_s\circ L\circ T_s$ satisfies $\tilde L\in\mathcal L(\ell^2)$.\end{lem}
\begin{proof}
 Take $x\in\ell^2$. Then by Lemma~\ref{transform},
 \[|\tilde L x|_{\ell^2}\leq|T^{-1}_s|_{\mathcal L(\ell^2_s,\ell^2)} |L|_{\mathcal L(\ell^2_s)} |T_s|_{\mathcal L(\ell^2,\ell^2_s)}|x|_{\ell^2}<\infty.\]
 By Lemma~\ref{transform} we conclude that $|\tilde L|_{\mathcal L(\ell^2)}\leq|L|_{\mathcal L(\ell^2_s)}$. Analogously, since
 $L=T_s\circ L\circ T^{-1}_s$, we have that $|\tilde L|_{\mathcal L(\ell^2)}=|L|_{\mathcal L(\ell^2_s)}.$
 \end{proof}
 \medskip
 
The tilde operator commutes with exponentiation.

\begin{lem}\label{commute} Let $L\in\mathcal L(\ell^2_s)$, then $\widetilde{L^n}=\tilde L^n$. \end{lem}

\begin{proof}
 We proceed by induction over $n$. For $n=1$, the proposition is a tautology. We assume the proposition true for $n$, let $x\in\ell^2$ and calculate
 \begin{align*}
 \tilde L^{n+1}x & = \tilde L\circ\tilde L^nx = T^{-1}_s\circ L\circ T_s\circ T^{-1}_s\circ L^n\circ T_s x
= T^{-1}_s\circ L^{n+1} \circ T_s x = \widetilde{L^{n+1}}x.
 \end{align*}
 \end{proof}

Recall from~\eqref{explicit} the explicit form of $w$ and $v$ and note that if $s\in(0,1)$
\begin{equation}|v|_{\ell^2_s}^2=\sum_{k\in\mathbb N}k^2s^k=\frac{s(1+s)}{(1-s)^3}<\infty.\label{vinl2s}\end{equation}
On the other hand, if $s>\left(\frac{1}{\alpha}-1\right)^2$,
\begin{equation}|w|_{\ell^2_{s^{-1}}}^2=\sum_{k\in\mathbb N}\left(\frac{1}{\alpha}-1\right)^{2(k-1)}s^{-k}=\frac{1}{s-(\tfrac{1}{\alpha}-1)^2}<\infty.\label{winel2sstar}\end{equation}
Therefore, for $s\in\left((\tfrac{1}{\alpha}-1)^2,1\right)$ we have that $v\in\ell^2_s$ and $w\in\ell^2_{s^{-1}}$.

\subsection{Toeplitz operators}

Before stating the main result of this section, we need to review some properties of Toeplitz operators.
Let $a=\{a_k\}_{k\in\mathbb Z}\in\ell^2(\mathbb C)$, that is, a double sequence of complex numbers such that $\sum_{k\in\mathbb Z}|a_k|^2<\infty$.
A Toeplitz operator $A$ defined by the double sequence $a\in\ell^2(\mathbb C)$ is an infinite matrix with the structure
\[A=\begin{pmatrix}
a_0 & a_{-1} & a_{-2} & \cdots \\
a_1 & a_0    & a_{-1} & \cdots \\
a_2 & a_1    & a_0    & \cdots \\
\vdots & \vdots & \ddots & \ddots \\
\end{pmatrix}.\]

The \emph{symbol} $\kappa\colon\{z\in\mathbb C:|z|= 1\}\to\mathbb C$ of a Toeplitz operator is defined by \[\kappa(z)=\sum_{k\in\mathbb Z}a_kz^k.\]
We recall Theorem 7.1 in \cite{spectra} that deals with spectra of Toeplitz operators.

\begin{teo}\cite[Theorem 7.1]{spectra}\label{spectrum}
Let $A$ be a Toeplitz operator. If $A$ has a continuous symbol $\kappa$, then its spectrum is given by the image of the unit circle under $\kappa$ together with all the points enclosed by this curve with non-zero winding number. \end{teo}

For fixed $\theta\in\mathbb R$, the operator $A(\theta)$ defined by~\eqref{toe} is by Proposition~\ref{abounded} in $\mathcal L(\ell^2_s)$. By Lemma~\ref{opnorm} the operator $\tilde A(\theta)$ is a Toeplitz operator in $\ell^2$ with its symbol $\kappa$ given by
\[\kappa(\zeta)=\frac{e^\theta}{\zeta\sqrt{s}}+1+e^\theta+\zeta\sqrt{s}.\]
Writing $\zeta=e^{i\varphi}$ as an element of the unit circle,
\[\kappa(e^{i\varphi})= 1+e^\theta +\left(\sqrt{s}+\frac{e^\theta}{\sqrt{s}}\right)\cos\varphi+\left(\sqrt{s}-\frac{e^\theta}{\sqrt{s}}\right)i\sin\varphi,\]
which we recognise as a parametrised ellipse centred at $1+e^\theta$, with major axis of length $\sqrt{s}+\frac{e^\theta}{\sqrt{s}}$ along the real line, and minor axis of length $|\sqrt{s}-\frac{e^\theta}{\sqrt{s}}|$. Therefore, the spectral radius is found when $z=1$ and 
\begin{equation}
  \label{eq:spr}
  \rho(\tilde A(\theta))=\kappa(1)=1+\sqrt{s}+e^\theta\left(1+\frac{1}{\sqrt{s}}\right).
\end{equation}

We now state the main result of this section: the upper bound for the cumulant generating function $\Lambda$.

\begin{prop}\label{UB}For $\Lambda$ defined by~\eqref{shortlambda}, 
an upper bound is 
\begin{equation*}
   \Lambda(\theta)\leq
 \begin{cases}
 \log\left(\frac{e^\theta}{1-\alpha}+\frac{1}{\alpha}\right)+\log c & \text{if }  -\infty<\theta\leq 2\log\left(\frac{1}{\alpha}-1\right),\\[2mm]
 \log\left(1+e^{\theta/2}\right)^2+\log c & \text{if } 2\log\left(\frac{1}{\alpha}-1\right)<\theta\leq -2\log\lambda_1,\\[2mm]
 \log\left(1+\lambda_1e^{\theta}\right)+\log\left(1+\frac{1}{\lambda_1}\right)+\log c & \text{if } -2\log\lambda_1<\theta<\infty. 
 \end{cases}
\end{equation*}
\end{prop}

\begin{proof}
Recall from~\eqref{vinl2s} that $v\in\ell^2_s$; by Proposition~\ref{abounded} $A(\theta)\in\mathcal L(\ell^2_s)$ and therefore $A(\theta)v\in\ell^2_s$. Also $w\in\ell^2_{s^{-1}}$   from~\eqref{winel2sstar}. Hence, by \eqref{shortlambda} and Cauchy-Schwarz,
\begin{align*}
\Lambda(\theta) & =  \lim_{n\rightarrow\infty}\tfrac{1}{n}\log(w^{\sf T}A(\theta)^nv)+\log c\\
 & \leq  \lim_{n\rightarrow\infty}\tfrac{1}{n} \log(|w|_{\ell^{2}_{s^{-1}}} |A(\theta)^nv|_{\ell^2_s})+\log c.\\
 & \phantom{= \lim_{n\rightarrow\infty}\tfrac{1}{n}\log(|T_s\circ\widetilde{A(\theta)^n}\circ T^{-1}_sv|_{\ell^2})+\log c}.\\[-13mm]
  \end{align*}
The norm of $w$ does not contribute to the limit since it does not depend on $n$. By Lemma~\ref{opnorm}, Lemma~\ref{transform}, and Lemma~\ref{commute} we can continue the previous estimate
\begin{align*}  
\phantom{\Lambda(\theta)} & =  \lim_{n\rightarrow\infty}\tfrac{1}{n}\log(|T_s\circ\widetilde{A(\theta)^n}\circ T^{-1}_sv|_{\ell^2})+\log c\\
 & \leq  \lim_{n\rightarrow\infty}\tfrac{1}{n}\log(|\tilde A(\theta)^n|_{\mathcal L(\ell^2)} |T_s^{-1}v|_{\ell^2})+\log c;
  \end{align*}
once again, the norm of $T_s^{-1}v$ does not contribute to the limit since it does not depend on $n$. 
We insert the factor $\tfrac{1}{n}$ to the logarithm by continuity and use the definition of spectral radius
 \begin{align*}
\phantom{\Lambda(\theta)} & =  \lim_{n\rightarrow\infty}\log(|\tilde A(\theta)^n|_{\mathcal L(\ell^2)}^{\frac{1}{n}}) +\log c\\
 & =  \log \rho(\tilde A(\theta))+\log c.\\
  & \phantom{= \lim_{n\rightarrow\infty}\tfrac{1}{n}\log(|T_s\circ\widetilde{A(\theta)^n}\circ T^{-1}_sv|_{\ell^2})+\log c}\\[-13mm]
 \end{align*}
We now use~\eqref{eq:spr} to find the spectral radius\\[-3mm]
\begin{align*}
\phantom{\Lambda(\theta)} & = \log\left[1+\sqrt{s}+e^\theta\left(1+\frac{1}{\sqrt{s}}\right)\right]+\log c.\\
 & \phantom{= \lim_{n\rightarrow\infty}\tfrac{1}{n}\log(|T_s\circ\widetilde{A(\theta)^n}\circ T^{-1}_sv|_{\ell^2})+\log c}\\[-13mm]
 \end{align*}
 Since the left hand side does not depend on $s$, it is a lower bound on the right hand side for $s$, so we take the infimum over the 
 interval  $\big((\tfrac{1}{\alpha}-1)^2,\tfrac{1}{\lambda_1^2}\big)$.
\[\Lambda(\theta)\leq\inf_{s\in\big((\frac{1}{\alpha}-1)^2,\frac{1}{\lambda_1^2}\big)}\log\left[1+\sqrt{s}+e^\theta\left(1+\frac{1}{\sqrt{s}}\right)\right]+\log c.\]
Given $\theta$, the value of $s$ that reaches the infimum of this function is given by
 \[s^{*}=\begin{cases}
  \left(\frac{1}{\alpha}-1\right)^2 & \text{if }   -\infty<\theta\leq 2\log\left(\frac{1}{\alpha}-1\right),\\[2mm]
 e^{\theta} & \text{if } \log\left(\frac{1}{\alpha}-1\right)^2<\theta\leq-2\log\lambda_1,\\[2mm]
 \frac{1}{\lambda_1^2} & \text{if } -2\log\lambda_1<\theta<\infty.
 \end{cases}\]
Plugging $s^*$ into the formula gives the result of the lemma.
\end{proof}

%Lower bound
%%%%%%%%%%%%%%%%%%%%%%%%%%%%%%%%%%%%%%%%%%
%%%%%%%%%%%%%%%%%%%%%%%%%%%%%%%%%%%%%%%%%%
%%%%%%%%%%%%%%%%%%%%%%%%%%%%%%%%%%%%%%%%%%
\section{Lower bound: A combinatorial approach}
In order to find the lower bound we use a completely different approach. We will focus on the powers of $e^\theta D+E$.

\begin{prop}\label{recursion}
 There exists a sequence of polynomials $f^n_{p,j}(\theta)$ 
in $e^\theta$ such that \begin{equation}\label{powerpj}(e^\theta D+E)^n=\sum_{p=1}^n\sum_{j=0}^pf^n_{p,j}(\theta)E^{p-j}D^j\end{equation}
and they can be defined recursively in two ways: Starting with
\[f^1_{1,0}(\theta)=1,\qquad f^1_{1,1}(\theta)=e^\theta,\]
the first characterisation for $f_{p,j}^n$ with $n>1$ is \\[-1mm]
\begin{equation}\label{rec1}
f_{p,j}^n(\theta)=\left\{\begin{array}{c@{\!\,\,}lll}
\sum_{k=1}^{n-1}f_{k,k}^{n-1}(\theta) & \text{if }n>1,& p=1, & j\in\{0,1\} \\[2mm]
f^{n-1}_{p-1,0}(\theta)+\sum_{k=p}^{n-1}f_{k,k-p+1}^{n-1}(\theta) & \text{if }n>1,& 1<p<n, & j=0 \\[2mm]
e^\theta f^{n-1}_{p-1,j-1}(\theta)+\sum_{k=p}^{n-1}f_{k,k-p+j}^{n-1}(\theta) & \text{if }n>1,& j\leq p<n, & j>0 \\[2mm]
f_{n-1,0}^{n-1}(\theta) & \text{if }n>1,& p=n, & j=0 \\[2mm]
e^\theta f_{n-1,j-1}^{n-1}(\theta) & \text{if }n>1,& p=n, & 0<j\leq p; \\
\end{array}\right.\end{equation}
and the second characterisation is 
\begin{equation}\label{rec2}f_{p,j}^n(\theta)=\left\{\begin{array}{c@{\!\,\,}lll}
\sum_{k=1}^{n-1}e^\theta f_{k,0}^{n-1}(\theta) & \text{if }n>1,& p=1, & j\in\{0,1\} \\[2mm]
f^{n-1}_{p-1,j}(\theta)+\sum_{k=p}^{n-1}e^\theta f_{k,j}^{n-1}(\theta) & \text{if }n>1,& 1<p<n, & 0\leq j<p \\[2mm]
e^\theta f^{n-1}_{p-1,p-1}(\theta)+\sum_{k=p}^{n-1}e^\theta f_{k,p-1}^{n-1}(\theta) & \text{if }n>1,& 1<p<n, & j=p \\[2mm]
f_{n-1,j}^{n-1}(\theta) & \text{if }n>1,& p=n, & 0\leq j<n \\[2mm]
e^\theta f_{n-1,n-1}^{n-1}(\theta) & \text{if }n>1,& p=n, & j=n. \\[1mm]
\end{array}\right.\end{equation}
\end{prop}

\begin{proof}
 We prove this by induction. For $n=1$ we have $(e^\theta D+E)^1=e^\theta D+E$ which settles the initial values $f^1_{1,0}(\theta)=1$ and $f^1_{1,1}(\theta)=e^\theta$. To find the recursion we assume the induction hypothesis:\[(e^\theta D+E)^n=\sum_{p=1}^n\sum_{j=0}^pf^n_{p,j}(\theta)E^{p-j}D^j\] and expand the next power. However, there are two ways we can use to expand, namely $(e^\theta D+E)^{n+1}=(e^\theta D+E)^n(e^\theta D+E)$ or $(e^\theta D+E)(e^\theta D+E)^n$. The former will give equation~\eqref{rec1}, the latter~\eqref{rec2}. The functions $f^n_{p,j}$ are all polynomials in $e^\theta$ because this holds for the induction hypothesis and the operations in the induction step are only multiplications and additions of polynomials with positive coefficients.
\end{proof}

We now state an auxiliary result. 
\begin{lem}\label{auxprop}
 For $n\geq2$ and $1\leq p\leq r\leq n-1$, we have the identity
 \begin{equation}
   \label{eq:binsum}
   \sum_{k=p}^rk\binom{n-1-k}{r-k}=\frac{np-pr+r}{n-r+1}\binom{n-p}{r-p}.
 \end{equation}
\end{lem}
\begin{proof}
First note that the cases $p=r$ and $r=n-1$ are easy to check directly. We prove the general case by induction over $n$. 
The case $n=2$ 
is again easy to see.
We now assume that~\eqref{eq:binsum} holds for fixed $n\geq2$ and all $1\leq p\leq r\leq n-1$.
\smallskip

To show the result for $n+1$ we may assume $1\leq p < r\leq n-1$, ignoring the easy cases settled at the beginning of the proof. 
Starting from the left hand side for $n+1$ and using the  induction hypothesis on the third equality,
\begin{align*}
\sum_{k=p}^rk\binom{n-k}{r-k} & = \sum_{k=p}^{r-1}k\left[\binom{n-k-1}{r-k-1}+\binom{n-k-1}{r-k} \right]+r\\
& =   \sum_{k=p}^{r-1}k\binom{n-k-1}{r-k-1}+\sum_{k=p}^{r}k\binom{n-k-1}{r-k}\\
& = \frac{np-p(r-1)+(r-1)}{n-(r-1)+1}\binom{n-p}{(r-1)-p}+\frac{np-pr+r}{n-r+1}\binom{n-p}{r-p}\\
& =   \frac{(n+1)p-pr+r}{(n+1)-r+1}\binom{(n+1)-p}{r-p}.
\end{align*}
\end{proof}

We now identify the coefficients $f^n_{p,p}(\theta)$, for $1\leq p \leq n$.

\begin{prop}\label{pp}
For the coefficients defined in Proposition~\ref{recursion}, 
\begin{equation}
  \label{eq:fpp}
f_{p,p}^n(\theta)=
\begin{cases}
\sum_{r=p}^{n-1}\frac{p}{n}\binom{n-p-1}{r-p}\binom{n}{r}e^{r\theta} & \text{ if }1\leq p<n,\\[2mm]
e^{n\theta} &  \text{ if }p=n.\\
\end{cases}
\end{equation}
\end{prop}

\begin{proof}
Putting $j=p$ in equation~\eqref{rec1} of Proposition~\ref{recursion} we get a simplified recursion:
\[f^1_{1,1}(\theta)=e^\theta\]
and
\begin{equation}\label{easyrec}f_{p,p}^n(\theta)=\left\{\begin{array}{cll}
\sum_{k=1}^{n-1}f_{k,k}^{n-1}(\theta) & \text{ if }n>1,& p=1 \\
e^\theta f^{n-1}_{p-1,p-1}(\theta)+\sum_{k=p}^{n-1}f_{k,k}^{n-1}(\theta) & \text{ if }n>1,& 1<p<n\\
e^\theta f_{n-1,n-1}^{n-1}(\theta) & \text{ if }n>1,& p=n. \\
\end{array}\right.\end{equation}
If $p=n$, it is easy to see by induction that
\begin{equation}\label{pjen}f^n_{n,n}(\theta)=e^{n\theta}.\end{equation}
Now, if $p<n$, we proceed again by induction. Here the base of induction has to be $n=2$. The recursion equation~\eqref{easyrec} gives
$f^2_{1,1}(\theta)=f^1_{1,1}(\theta)=e^\theta$, as required by formula~\eqref{eq:fpp}. 
We now assume 
that~\eqref{eq:fpp} holds for fixed $n\geq 2$ and all $p<n$. 
We first consider the branch of~\eqref{easyrec}, referring to the case~$p=1$. 
Using~\eqref{pjen} and the induction hypothesis we obtain
\begin{align*}
   f^{n+1}_{1,1}(\theta) & =  \sum_{k=1}^nf^n_{k,k}(\theta)
 =  \sum_{k=1}^{n-1}f^n_{k,k}(\theta)+f^n_{n,n}(\theta)
 =  \sum_{k=1}^{n-1}\sum_{r=k}^{n-1}\frac{k}{n}\binom{n-k-1}{r-k}\binom{n}{r}e^{r\theta}+e^{n\theta},
\phantom{XXXX}
\end{align*}
changing the order of summation and using Lemma~\ref{auxprop} gives
\begin{align*}
\phantom{f^{n+1}_{1,1}}
& =  \sum_{r=1}^{n-1}\sum_{k=1}^{r}\frac{k}{n}\binom{n-k-1}{r-k}\binom{n}{r}e^{r\theta}+e^{n\theta}
=  \sum_{r=1}^{n-1}\frac{1}{n}\binom{n}{r}e^{r\theta}\sum_{k=1}^{r}k\binom{n-k-1}{r-k}+e^{n\theta}\\
& =  \sum_{r=1}^{n-1}\frac{1}{n}\binom{n}{r}e^{r\theta}\frac{n}{n-r+1}\binom{n-1}{r-1}+e^{n\theta}
=  \sum_{r=1}^{n}\frac{1}{n+1}\binom{n-1}{r-1}\binom{n+1}{r}e^{r\theta}.\\
&\phantom{=\sum_{r=1}^{n-1}\frac{1}{n}\binom{n}{r}e^{r\theta}\frac{n}{n-r+1}\binom{n-1}{r-1}+e^{n\theta}.}\\[-18mm]
\end{align*}
Since this is the result required by the induction step, the case $p=1$ is settled. We can therefore turn our attention to the 
remaining branch of~\eqref{easyrec}, covering $1<p\leq n$. We obtain from the induction hypothesis
\begin{align*}
f^{n+1}_{p,p}(\theta) & = e^\theta f^n_{p-1,p-1}(\theta)+\sum_{k=p}^nf^n_{k,k}(\theta)\\
& =  e^\theta\sum_{r=p-1}^{n-1}\frac{p-1}{n}\binom{n-p}{r-p+1}\binom{n}{r}e^{r\theta} +\sum_{k=p}^{n-1}\sum_{r=k}^{n-1}\frac{k}{n}\binom{n-k-1}{r-k}\binom{n}{r}e^{r\theta}+e^{n\theta};\\[-5mm]
\end{align*}
changing the summation order and grouping the terms by powers of $e^\theta$ yields
\begin{align*}
\phantom{f^{n+1}_{p,p}(\theta)}& =  \sum_{r=p}^{n}\frac{p-1}{n}\binom{n-p}{r-p}\binom{n}{r-1}e^{r\theta} +\sum_{r=p}^{n-1}\sum_{k=p}^{r}\frac{k}{n}\binom{n-k-1}{r-k}\binom{n}{r}e^{r\theta}+e^{n\theta}\\
& = \sum_{r=p}^{n-1}\left[\frac{p-1}{n}\binom{n-p}{r-p}\binom{n}{r-1} +\sum_{k=p}^{r}\frac{k}{n}\binom{n-k-1}{r-k}\binom{n}{r}\right]e^{r\theta}\\
& \phantom{= +}+\left[\frac{p-1}{n}\binom{n-p}{n-p}\binom{n}{n-1} +1\right]e^{n\theta};\\
& \phantom{=  e^\theta\sum_{r=p-1}^{n-1}\frac{p-1}{n}\binom{n-p}{r-p+1}\binom{n}{r}e^{r\theta} +\sum_{k=p}^{n-1}\sum_{r=k}^{n-1}\frac{k}{n}\binom{n-k-1}{r-k}\binom{n}{r}e^{r\theta}+e^{n\theta}}\\
\end{align*}
\vspace{-2cm}

simplifying and using Lemma~\ref{auxprop} in the second line gives
\begin{align*}
\phantom{f^{n+1}_{p,p}(\theta)}& =  \sum_{r=p}^{n-1}\left[\frac{r(p-1)}{n-r+1}\binom{n-p}{r-p}+\sum_{k=p}^{r}k\binom{n-k-1}{r-k}\right]\frac{1}{n}\binom{n}{r}e^{r\theta}+pe^{n\theta}\\
& = \sum_{r=p}^{n-1}\left[\frac{r(p-1)}{n-r+1}\binom{n-p}{r-p}+\frac{np-pr+r}{n-r+1}\binom{n-p}{r-p}\right]\frac{1}{n}\binom{n}{r}e^{r\theta}+pe^{n\theta}\phantom{XXXX}\\
\end{align*}
\begin{align*}
& = \sum_{r=p}^{n-1}\left[\frac{r(p-1)}{n-r+1}+\frac{np-pr+r}{n-r+1}\right]\frac{1}{n}\binom{n-p}{r-p}\binom{n}{r}e^{r\theta}+pe^{n\theta};\\
&\phantom{=  e^\theta\sum_{r=p-1}^{n-1}\frac{p-1}{n}\binom{n-p}{r-p+1}\binom{n}{r}e^{r\theta} +\sum_{k=p}^{n-1}\sum_{r=k}^{n-1}\frac{k}{n}\binom{n-k-1}{r-k}\binom{n}{r}e^{r\theta}+e^{n\theta}}
\end{align*}
\vspace{-1.5cm}

grouping that last term with the rest of the terms in the sum finally results in 
\begin{align*}
\phantom{f^{n+1}_{p,p}(\theta)}& = \sum_{r=p}^{n}\frac{p}{n+1}\binom{n-p}{r-p}\binom{n+1}{r}e^{r\theta}.\\
&\phantom{=  e^\theta\sum_{r=p-1}^{n-1}\frac{p-1}{n}\binom{n-p}{r-p+1}\binom{n}{r}e^{r\theta} +\sum_{k=p}^{n-1}\sum_{r=k}^{n-1}\frac{k}{n}\binom{n-k-1}{r-k}\binom{n}{r}e^{r\theta}+e^{n\theta}}\\[-2cm]
\end{align*}
\end{proof}

\begin{prop}\label{symm} For all $n\in\mathbb N$, $1\leq p \leq n$ and $0\leq j \leq p$ we have the symmetry
 \[f^n_{p,j}(\theta)=e^{n\theta}f^n_{p,p-j}(-\theta).\]
\end{prop}
\begin{proof}
 Defining the polynomials $g^n_{p,j}(\theta):=e^{n\theta}f^n_{p,p-j}(-\theta)$, we can write 
\begin{equation*}
  f^n_{p,j}(\theta)=e^{n\theta}g^n_{p,p-j}(-\theta),
\end{equation*}
and by the definition~\eqref{powerpj} of the polynomials $f^n_{p,j}(\theta)$ we obtain by changing the summation index
\begin{align}
   (e^\theta D+E)^n & =  \sum_{p=1}^n\sum_{j=0}^pf^n_{p,j}(\theta)E^{p-j}D^j
 = \sum_{p=1}^n\sum_{j=0}^pe^{n\theta}g^n_{p,p-j}(-\theta)E^{p-j}D^j \notag\\
 &=  \sum_{p=1}^n\sum_{j=0}^pe^{n\theta}g^n_{p,j}(-\theta)E^{j}D^{p-j}.
\label{eq:rec-g}  
\end{align}
Evaluating this expression for $n=1$, we find \[e^\theta D+E=e^\theta g^1_{1,0}(-\theta)D+e^\theta g^1_{1,1}(-\theta)E,\] and hence $g^1_{1,1}(\theta)=e^\theta$ and $g^1_{1,0}(\theta)=1$. Next we find a recursive relation for these polynomials by expanding
 and employing~\eqref{eq:rec-g}, 
\begin{align*}
   (e^\theta D+E)^{n+1} & =  (e^\theta D+E)^{n}(e^\theta D+E)\\
 & =  \left(\sum_{p=1}^n\sum_{j=0}^pe^{n\theta}g^n_{p,j}(-\theta)E^{j}D^{p-j}\right)(e^\theta D+E)
  \end{align*}
and equating the coefficients to 
\[ (e^\theta D+E)^{n+1}=\sum_{p=1}^{n+1}\sum_{j=0}^pe^{n\theta}g^n_{p,j}(-\theta)E^{j}D^{p-j}\]
we find that the polynomials $g^n_{p,j}$ satisfy the following recursion:
\[g^1_{1,0}(\theta)=1,\qquad g^1_{1,1}(\theta)=e^\theta\]
\[g_{p,j}^n(\theta)=\left\{\begin{array}{clll}
\sum_{k=1}^{n-1}e^\theta g_{k,0}^{n-1}(\theta) & \text{ if }n>1,& p=1, & j\in\{0,1\} \\
g^{n-1}_{p-1,j}(\theta)+\sum_{k=p}^{n-1}e^\theta g_{k,j}^{n-1}(\theta) & \text{ if }n>1,& 1<p<n, & 0\leq j<p \\
e^\theta g^{n-1}_{p-1,p-1}(\theta)+\sum_{k=p}^{n-1}e^\theta g_{k,p-1}^{n-1}(\theta) & \text{ if }n>1,& 1<p<n, & j=p \\
g_{n-1,j}^{n-1}(\theta) & \text{ if }n>1,& p=n, & 0\leq j<n \\
e^\theta g_{n-1,n-1}^{n-1}(\theta) & \text{ if }n>1,& p=n, & j=n. \\
\end{array}\right.\]
This is the recursion equation~\eqref{rec2} of Proposition~\ref{recursion}, and hence \smash{$f^n_{p,j}(\theta)=g^n_{p,j}(\theta)$}.
Thus from the definition of $g^n_{p,j}(\theta)$ we conclude that
$f^n_{p,j}(\theta)=e^{n\theta}f^n_{p,p-j}(-\theta)$ as claimed.
\end{proof}

From the symmetry rule of Proposition~\ref{symm} we obtain a simple characterisation of the coefficients $f^n_{p,0}(\theta)$.

\begin{cor}\label{fnp0}
  \[f_{p,0}^n(\theta)=\left\{\begin{array}{cl}
\displaystyle\sum_{r=1}^{n-p}\frac{p}{n}\binom{n-p-1}{r-1}\binom{n}{r}e^{r\theta} & \text{ if }1\leq p<n\\
1 &  \text{ if }p=n.\\
\end{array}\right.\]
\end{cor}
\begin{proof}
The result follows by combining Proposition~\ref{symm} for $j=p$, 
and Proposition~\ref{pp}.
\end{proof}

We plug the expansion~\eqref{powerpj} into equation~\eqref{shortlambda} and obtain
\begin{equation}\label{lambdapj}\Lambda(\theta) = \lim_{n\rightarrow\infty}\frac{1}{n}\log\left\{\sum_{p=1}^n\sum_{j=0}^p f^n_{p,j}(\theta) w^{\sf T} E^{p-j} D^jv\right\}+\log c.\end{equation}
Since all terms are positive, we can find lower bounds by considering only the terms for which $j=0$ or $j=p$. This is the 
content of the next couple of results.

\begin{prop}\label{LB1}For the cumulant generating function $\Lambda$   defined by~\eqref{shortlambda},
 \[\Lambda(\theta)\geq
 \begin{cases}
 2\log\left(1+e^{\theta/2}\right)+\log c & \text{ if } \theta\leq -2\log\lambda_1,\\
 \log\left(1+\lambda_1e^{\theta}\right)+\log\left(1+\frac{1}{\lambda_1} \right)+\log c & \text{ if }\theta >-2\log\lambda_1.\\
 \end{cases}\]
\end{prop}
\begin{proof}
From the explicit form of $D$, $w$, and $v$ in~\eqref{explicit} it can be shown by induction that
\begin{equation}
  \label{eq:wDv}
   w^{\sf T}D^pv=(1+\lambda_1)^p\frac{\alpha}{\lambda_1-\lambda_2} \left\{\frac{\lambda_1}{\alpha-\lambda_1(1-\alpha)} - \left(\frac{1+\lambda_2}{1+\lambda_1} \right)^p \frac{\lambda_2}{\alpha-\lambda_2(1-\alpha)} \right\}.
\end{equation} 
We consider only the terms for which $j=p$ in equation~\eqref{lambdapj} and note that the expression in the square parenthesis in~\eqref{eq:wDv} vanishes in the limit taken in~\eqref{lambdapj}. Hence, using the Laplace principle and Proposition~\ref{pp},
\begin{align*}
\Lambda(\theta) &\geq \lim_{n\rightarrow\infty}\frac{1}{n}\log\left\{\sum_{p=1}^n f^n_{p,p}(\theta)(1+\lambda_1)^p\right\}+\log c\\
&= \lim_{n\rightarrow\infty}\sup_{1\leq p\leq r\leq n} \frac{1}{n}\log\left\{ \frac{p}{n}\binom{n-p-1}{r-p}\binom{n}{r}e^{r\theta}(1+\lambda_1)^p\right\}+\log c.
  \end{align*}
We now use Stirling's formula and a change of variables $\varepsilon=\frac{r}{n}$ and $\delta=\frac{p}{n}$ to obtain
\begin{align*}
\Lambda(\theta) &\geq  {\displaystyle \sup_{0< \delta\leq \varepsilon\leq 1}\left[ \lim_{n\rightarrow\infty}\frac{1}{n}\log\left\{ \delta\binom{n(1-\delta)}{n(\varepsilon-\delta)}\binom{n}{n\varepsilon}\right\}+\varepsilon\theta+\delta\log(1+\lambda_1)\right]+\log c}\\
& =  {\displaystyle\sup_{0<\varepsilon\leq1}\left[\sup_{0<\delta\leq\varepsilon}\left\{-(\varepsilon-\delta)
\log(\varepsilon-\delta)+(1-\delta)\log(1-\delta)+\delta\log(1+\lambda_1)\right\}\right.}\\
&  {\displaystyle \left.\phantom{\sup_{0<\varepsilon\leq1}123} +\varepsilon\theta-2(1-\varepsilon)\log(1-\varepsilon)-\varepsilon\log\varepsilon\right]+\log c}.
  \end{align*}
The inner problem, when $\varepsilon$ is fixed, is solved by choosing $\delta^{\max}=0$, if  $\varepsilon\leq\frac{1}{1+\lambda_1}$, and $\delta^{\max}=\frac{(1+\lambda_1)\varepsilon-1}{\lambda_1}$, if $\varepsilon >\frac{1}{1+\lambda_1}$.
So now we have
\[\begin{array}{rcl}
\Lambda(\theta) &\geq &  {\displaystyle \max\left\{ \sup_{0<\varepsilon\leq\frac{1}{1+\lambda_1}}\left[\varepsilon\theta-2\varepsilon\log\varepsilon-2(1-\varepsilon)\log(1-\varepsilon) \right],\right.}\\
& & {\displaystyle\left.\phantom{\max1} \sup_{\frac{1}{1+\lambda_1}<\varepsilon\leq1}\left[\varepsilon\theta-\varepsilon\log\varepsilon-(1-\varepsilon) \log[\lambda_1(1-\varepsilon)]+\log(1+\lambda_1)   \right]\right\}+\log c.}\\ 
  \end{array}\]
This problem is solved by choosing
\[\varepsilon^{\max}=\left\{\begin{array}{cl}
   \frac{e^{\theta/2}}{1+e^{\theta/2}} & \text{if } \theta\leq -2\log\lambda_1\\
 & \\
    \frac{\lambda_1e^{\theta}}{1+\lambda_1e^{\theta}} &  \text{if } \theta > -2\log\lambda_1. \\
  \end{array}\right.\]

Plugging this value of $\varepsilon^{\max}$ yields the result of the proposition.
\end{proof}

\begin{prop}\label{LB2} For the cumulant generating function $\Lambda$ defined by~\eqref{shortlambda}, 
 \[\Lambda(\theta)\geq
\begin{cases}
\log\left(\frac{e^{\theta}}{1-\alpha}+\frac{1}{\alpha}\right)+\log c & \text{ if }\theta \leq 2\log\left(\frac{1}{\alpha}-1\right),\\
2\log\left(1+e^{\theta/2}\right)+\log c & \text{ if } \theta  >  2\log\left(\frac{1}{\alpha}-1\right).\\
\end{cases}\] 
\end{prop}

\begin{proof}
 We follow the same technique as in the previous proposition, but now consider only those values for which $j=0$ in~\eqref{powerpj}. Hence, using Corollary \ref{fnp0}
\[\begin{array}{rcl}
\Lambda(\theta) &\geq &  {\displaystyle \lim_{n\rightarrow\infty}\frac{1}{n}\log\left\{\sum_{p=1}^n f^n_{p,0}(\theta)\frac{1}{\alpha^p}\right\}+\log c}\\
&= & {\displaystyle \lim_{n\rightarrow\infty}\sup_{\substack{1 \leq p \leq n \\ 0\leq r\leq n-p}}\frac{1}{n} \log\left\{ \frac{p}{n}\binom{n-p-1}{r-1}\binom{n}{r}e^{r\theta}\alpha^{-p}\right\}+\log c}\\
  \end{array}\]
With the same change of variables as before,  $\varepsilon=\frac{r}{n}$ and $\delta=\frac{p}{n}$,  we have
\[\begin{array}{rcl}
\Lambda(\theta) &\geq & \displaystyle\sup_{\substack{0 < \delta \leq 1 \\ 0\leq \varepsilon \leq 1-\delta}}\left[ \lim_{n\rightarrow\infty}\frac{1}{n}\log\left\{ \delta\binom{n(1-\delta)}{n\varepsilon}\binom{n}{n\varepsilon}\right\}+\varepsilon\theta-\delta\log\alpha\right]+\log v\\
& = & \displaystyle\sup_{\substack{0 \leq \delta \leq 1 \\ 0\leq \varepsilon\leq 1-\delta}}\left\{(1-\delta)\log(1-\delta)-(1-\delta-\varepsilon)\log(1-\delta-\varepsilon)-\delta\log\alpha\right.\\
& & \phantom{\sup_{\substack{0 \leq \delta \leq 1 \\ 0\leq \varepsilon\leq 1-\delta}}} \left. -2\varepsilon\log\varepsilon-(1-\varepsilon)\log(1-\varepsilon)+\varepsilon\theta\right\} +\log c. \\
  \end{array}\]
Splitting the problem in two, for a fixed $\varepsilon$, we find the optimal $\delta$ as 
$\delta^{\max}=0$, if $1-\alpha\leq \varepsilon\leq 1$, and \smash{$\delta^{\max}=1-\frac{\varepsilon}{1-\alpha}$},
if $0\leq \varepsilon < 1-\alpha$.
And the remaining problem is solved by choosing the optimum $\varepsilon$ as
\[\varepsilon^{\max}=\begin{cases}
   \frac{e^{\theta/2}}{1+e^{\theta/2}} & \text{if } \theta\leq 2\log\left(\frac{1}{\alpha}-1\right),\\
    \frac{\alpha}{\alpha+e^{-\theta}(1-\alpha)} &  \text{if } \theta > 2\log\left(\frac{1}{\alpha}-1\right). \\
  \end{cases}\]
With these values of $\varepsilon$ we get the desired result.
\end{proof}

\begin{cor}\label{LB}For $\Lambda$ defined by~\eqref{shortlambda},
  \[\Lambda(\theta)\geq
 \begin{cases}
 \log\left(\frac{e^\theta}{1-\alpha}+\frac{1}{\alpha}\right)+\log c & \text{if }-\infty < \theta \leq 2\log\left(\frac{1}{\alpha}-1\right),\\[2mm]
 2\log\left(1+e^{\theta/2}\right)+\log c & \text{if } 2\log\left(\frac{1}{\alpha}-1\right)  <  \theta  \leq  -2\log\lambda_1,\\[2mm]
 \log\left(1+\lambda_1e^{\theta}\right)+\log\left(1+\frac{1}{\lambda_1}\right)+\log c & \text{if }-2\log\lambda_1<\theta<\infty.
 \end{cases}\]
\end{cor}

\begin{proof}
  The bounds from Propositions~\ref{LB1} and~\ref{LB2} are the same in the   interval $2\log\left(\frac{1}{\alpha}-1\right)\leq\theta\leq -2\log\lambda_1$. In   the other intervals, a comparison of the bounds establishes the   claim.
\end{proof}

%Proof of main result
%%%%%%%%%%%%%%%%%%%%%%%%%%%%%%%%%%%%%%%%%%
%%%%%%%%%%%%%%%%%%%%%%%%%%%%%%%%%%%%%%%%%%
%%%%%%%%%%%%%%%%%%%%%%%%%%%%%%%%%%%%%%%%%%
\section{The rate function}

Summarising, we have the following result.

\begin{cor}\label{longlambda}For the cumulant generating function $\Lambda$ defined by~\eqref{shortlambda},
 \begin{equation}\label{explicitlambda}\Lambda(\theta)=
\begin{cases}
\log\left(\frac{e^\theta}{1-\alpha}+\frac{1}{\alpha}\right)+\log c & \text{if }-\infty < \theta \leq 2\log\left(\frac{1}{\alpha}-1\right),\\[2mm]
2\log\left(1+e^{\theta/2}\right)+\log c & \text{if } 2\log\left(\frac{1}{\alpha}-1\right)  <  \theta  \leq  -2\log\lambda_1,\\[2mm]
\log\left(1+\lambda_1e^{\theta}\right)+\log\left(1+\frac{1}{\lambda_1}\right)+\log c & \text{if }-2\log\lambda_1<\theta<\infty.
\end{cases}\end{equation}
\end{cor}

\begin{proof}
 This follows from the fact that the upper and lower bounds from Proposition~\ref{UB} and Corollary~\ref{LB}, respectively, are the same.  
\end{proof}

Finally we have the necessary tools to prove Theorem~\ref{main}.

\begin{proof}
The rate function in the case \emph{(a)} $\alpha\leq {1}/{2}$ and $\rho<1-\alpha$ is known from Cram\'er's theorem, see e.g. Exercise 2.2.23 (b) in \cite{zeitouni}. For the case \emph{(c)}, $\alpha>{1}/{2}$ and $\tfrac{1}{2}<\rho<\alpha$, we show that the function 
$\Lambda$ defined by~\eqref{shortlambda}, given explicitly in Corollary~\ref{longlambda}, satisfies 
the hypotheses of the G\"artner-Ellis theorem~\ref{GE}.
Note that $\Lambda$ is defined for all real numbers. An evaluation at the boundaries of the domains
gives   \[\Lambda\left(2\log(\tfrac{1}{\alpha}-1)\right)=-2\log\alpha+\log c =\lim_{h\to0^{+}}\Lambda\left(2\log(\tfrac{1}{\alpha}-1)-h\right)\]
and
\[\Lambda\left(-2\log\lambda_1\right)=2\log(1+\tfrac{1}{\lambda_1})+\log c =\lim_{h\to0^{+}}\Lambda\left(-2\log\lambda_1+h\right),\]
which implies that it is continuous in all $\mathbb R$. Moreover,
\[\lim_{h\to0^{\pm}}\frac{\Lambda\left(2\log(\tfrac{1}{\alpha}-1)+h\right)-\Lambda\left(2\log(\tfrac{1}{\alpha}-1)\right)}{h}=1-\alpha\]
and
\[\lim_{h\to0^{\pm}}\frac{\Lambda\left(-2\log\lambda_1+h\right)-\Lambda\left(-2\log\lambda_1\right)}{h}=\frac{1}{1+\lambda_1}.\]
Therefore, $\Lambda$ is differentiable in $\mathbb R$. By the G\"artner-Ellis theorem we  find the rate function, \[I(z)=\sup_{\theta\in\mathbb R}\left\{z\theta-\Lambda(\theta)\right\}
\qquad \mbox{ for $z\in(0,1)$.}\] 
For fixed $z\in(0,1)$ the function $\theta\mapsto z\theta-\Lambda(\theta)$ is well-defined, continuous, and differentiable in all $\mathbb R$. It is also a concave function and hence the maximum is reached at a value of $\theta$ where the derivative vanishes. Since
\[\frac{d}{d\theta}\left(z\theta-\Lambda(\theta)\right)=
\begin{cases}
z-\frac{\alpha e^\theta}{\alpha e^\theta +1-\alpha}& \text{if } -\infty<\theta\leq\log\left(\frac{1}{\alpha}-1\right)^2,\\[2mm]
z-\frac{ e^{\theta/2}}{1+ e^{\theta/2}} & \text{if } \log\left(\frac{1}{\alpha}-1\right)^2<\theta\leq -2\log\lambda_1,\\[2mm]
z-\frac{\lambda_1 e^{\theta}}{1+ \lambda_1 e^{\theta}} & \text{if }  -2\log\lambda_1<\theta<\infty,\\
\end{cases}\]
we get $\frac{d}{d\theta}\left(z\theta-\Lambda(\theta)\right)=0$ if and only if
\[\begin{array}{clcl}
& \theta=\log\frac{z(1-\alpha)}{\alpha(1-z)}& \text{and} & \theta\leq2\log\left(\frac{1}{\alpha}-1\right),\\[2mm]
\text{or}&\theta=2\log\frac{z}{1-z}& \text{and}& 2\log\left(\frac{1}{\alpha}-1\right)<\theta\leq-2\log\lambda_1,\\[2mm]
\text{or}&\theta=\log\frac{z}{\lambda_1(1-z)}& \text{and}& \theta>-2\log\lambda_1.\\
\end{array}\]
This means that 
\[\begin{array}{lcl}
\theta=\log\frac{z(1-\alpha)}{\alpha(1-z)}& \Leftrightarrow & 0<z<1-\alpha,\\[2mm]
\theta=2\log\frac{z}{1-z}& \Leftrightarrow & 1-\alpha < z\leq\frac{1}{1+\lambda_1},\\[2mm]
\theta=\log\frac{z}{\lambda_1(1-z)}& \Leftrightarrow & \frac{1}{1+\lambda_1}<z<1.\\
\end{array}\]
Since the value of $\theta$ that satisfies $\tfrac{d}{d\theta}\left(z\theta-\Lambda(\theta)\right)=0$ is unique it must be the maximum. By plugging in this value in~\eqref{explicitlambda}, we reach the desired result.
The remaining case \emph{(b)} is obtained simply by taking the limit of case \emph{(c)} as $c\to\tfrac{1}{4}$.
\end{proof}
\medskip

An alternative approach to the problems studied here is based on translation into a random walk problem. Rewriting the infinite matrix $e^\theta D +E$ as $2(1+e^\theta)Q_\theta$, one can see that $Q_\theta$ can be interpreted as the probability transition matrix of a discrete time random walk~$\{X_n\}_{n\geq0}$ on $\mathbb N_0$ killed at the origin. By the form of the vector $w^{\sf T}$, shown in~\eqref{explicit}, the product $w^{\sf T}Q_\theta^n$ is, up to a normalising constant, the sub-probability distribution of the $n$-th step of the random walk starting from a geometric distribution with parameter $2-\tfrac{1}{\alpha}$. If $T$ denotes the killing time of the random walk, then~\eqref{explicitlambda} becomes
\begin{align*}\Lambda(\theta) & = \lim_{n\to\infty}\tfrac{1}{n}\log w^{\sf T}(e^\theta D+E)^nv+\log c\\
& = \log\big(2c(1+e^\theta)\big) + \lim_{n\to\infty}\tfrac{1}{n}\log\E\big[(\lambda_1^{X_n}-\lambda_2^{X_n})\indicator{\{T>n\}}\big], 
\end{align*}
using the form of $v$ given in~\eqref{explicit}. The latter rate may be evaluated using a functional large deviation principle for the random walk.  This alternative
method hinges on the availability of the large deviation result and the resolvability of the variational problems that come out of its application, which may be more 
complicated than our original approach. In principle, however, the method seems suited to  not only reveal large deviations for the asymptotic density, but also for a 
density profile depending on a macroscopic space variable, as done by Derrida \emph{et al.} in \cite{DLS} in the case of the finite TASEP. 
\medskip

By contrast, the technique presented in this paper is more elementary and direct. It therefore seems to be more flexible, giving hope to produce large deviation principles
for a range of other systems with spatially correlated distributions given by a matrix product representation. In particular there are several natural variants even of the example of a semi-infinite TASEP considered in this paper: For example, we would be interested in identifying a large deviation principle for the semi-infinite TASEP process starting from a configuration given by a value of $\rho$ in Theorem~\ref{liggettteo} with 
$\rho>\alpha>\frac12$, which was so far excluded for technical reasons.
It also seems feasible to generalise the large deviation principle for the overall density to a large deviation principle for a density profile. Finally, it would be interesting to analyse the non-stationary TASEP
using a time dependent matrix representation, as given by Stinchcombe and Sch\"utz in~\cite{SS}. 
\bigskip

%\newpage

\textbf{Acknowledgements:} HGD has been supported by CONACYT with CVU 302663, PM by EPSRC Grant EP/K016075/1. JZ has been partially supported by the Leverhulme Trust, Grant RPG-2013-261,
EPSRC, Grant EP/K027743/1 and a Royal Society Wolfson Research Merit Award. We also thank Rob Jack for his useful comments.

% The bibliography
%\bibliographystyle{apalike}
\bibliographystyle{plain}
\bibliography{biblio}
%\addcontentsline{toc}{chapter}{Bibliography}

\end{document}